\documentclass[10pt]{article}
\usepackage[utf8]{inputenc}
\title{
	Automatic complexity of Fibonacci and Tribonacci words
}
\author{Bj{\o}rn Kjos-Hanssen\thanks{
		This work was partially supported by grants from the
		Simons Foundation (\#315188 and \#704836 to Bj\o rn Kjos-Hanssen) and Decision Research Corporation (\emph{Automatic Complexity of Fibonacci arrays}, University of Hawai\textquoteleft i Foundation Account \#129-4770-4).
	}}
\usepackage[normalem]{ulem}
\usepackage{amsmath,amsthm,amssymb,mathrsfs}
\usepackage{graphicx}
\usepackage{tikz}
\usepackage{booktabs} 
\usepackage{enumerate}
\usetikzlibrary{matrix}
\usepackage{lscape}
\usetikzlibrary{automata,positioning}
\usepackage{caption}
\usepackage{subcaption}
\usepackage{multicol}
\usepackage{float}
\usepackage{chngpage}
\usetikzlibrary{decorations.pathmorphing}
\tikzset{snake it/.style={decorate, decoration=snake}}
\usepackage{tikz-qtree}
\usepackage{mathtools}
\usepackage{tcolorbox}
\usepackage{collectbox}
\usepackage{pgf}
\usepackage{pgfplots}
\usepackage{longtable}
\usepackage{rotating}
\usepackage{hyperref}

\usepackage[all]{xy}

\newtheorem{thm}{Theorem}
\newtheorem{lem}[thm]{Lemma}

\newtheorem{rem}[thm]{Remark}

\theoremstyle{definition}
\newtheorem{df}[thm]{Definition}

\newcommand{\mt}{\mathtt}
\newcommand{\abs}[1]{\lvert#1\rvert}

\newcommand{\restrict}{\upharpoonright}
\newcommand*\ci[1]{\langle #1\rangle}
\newcommand*\di[1]{\| #1 \|}

\begin{document}
	\maketitle
	\begin{abstract}
		For a complexity function $C$, the lower and upper $C$-complexity rates of an infinite word $\mathbf{x}$ are
		\[
			\underline{C}(\mathbf x)=\liminf_{n\to\infty} \frac{C(\mathbf{x}\restrict n)}n,\quad
			 \overline{C}(\mathbf x)=\limsup_{n\to\infty} \frac{C(\mathbf{x}\restrict n)}n
		\]
		respectively. Here $\mathbf{x}\restrict n$ is the prefix of $x$ of length $n$.
		We consider the case $C=\mathrm{A_N}$, the nondeterministic automatic complexity.
		If these rates  are strictly between 0 and $1/2$, we call them intermediate.
		Our main result is that words having intermediate $\mathrm{A_N}$-rates exist, viz.~the infinite Fibonacci and Tribonacci words.
	\end{abstract}

	\section{Introduction}
		The automatic complexity of Shallit and Wang~\cite{MR1897300} is
		the minimal number of states of an automaton accepting only a given word among its equal-length peers.
		This paper continues a line of investigation into the automatic complexity of particular words of interest such as
		\begin{itemize}
			\item maximal length sequences for linear feedback shift registers~\cite{MR3828751},
			\item overlap-free and almost square-free words~\cite{MR3386523}, and
			\item random words~\cite{kjoshanssen2019incompressibility}.
		\end{itemize}
		All these examples have high complexity: to be precise, they have maximal automatic complexity rate (Definition~\ref{rate}).
		On the other hand, a periodic word has low complexity and a rate of 0.
		In the present paper we give the first examples of infinite words with intermediate automatic complexity rate: the infinite Fibonacci and Tribonacci words. 

		Automatic complexity is an automata-based and length-conditional analogue of
		Sipser's CD complexity~\cite{Sipser:1983:CTA:800061.808762} which is in turn a computable analogue of the noncomputable
		Kolmogorov complexity.
		The nondeterministic case was taken up by Hyde and Kjos-Hanssen~\cite{MR3386523}.
		We recall our basic notions. Let $\abs{x}$ denote the length of a word $x$.
		\begin{df}[{\cite{MR1897300}}]
			The \emph{nondeterministic automatic complexity} $\mathrm{A_N}(x)$ of a word $x$ is
			the minimal number of states of a nondeterministic finite automaton $M$ (without $\epsilon$-transitions)
			such that $M$ accepts $x$ and moreover
			there is only one accepting path in $M$ of length $\abs{x}$.
			We let $\mathrm{A}^-$ denote the deterministic but non-total automatic complexity, defined as follows:
			automata are required to be deterministic, but the transition functions need not be total: there does not need to be a transition for every symbol at every state.
		\end{df}
		Shallit and Wang's original automatic complexity $\mathrm{A}(x)$ does have the totality requirement.

	\subsection{Fibonacci and Tribonacci words}
		\begin{df}[$k$-bonacci numbers]
			For $k\ge 2$ and $n\ge 0$, the $n$th $k$-bonacci number $w_n=w_n^{(k)}$ is defined by $w_n=0$ if $n\le k-1$,
			$w_k=1$,
			and $w_n=\sum_{i=n-k}^{n-1}w_i$ for $n\ge k+1$.
		\end{df}
		In particular, the Fibonacci numbers are the $2$-bonacci numbers.

		Let $\Sigma_k=\{\mt a_0,\dots,\mt a_{k-1}\}$ be an alphabet of cardinality $k\ge 0$.
		We shall denote specific symbols in the alphabet as $\mt 0=\mt a_0$, $\mt 1 = \mt a_1$ and so on.
		\begin{df}[$k$-bonacci words]
			We define the $k$-morphism $\varphi_k:\Sigma_k\to \Sigma_k^*$ by
			\begin{eqnarray*}
				\varphi_k(\mt a_i)		&=& \mt a_0 \mt a_{i+1},\quad 0\le i\le k-2,\\
				\varphi_k(\mt a_{k-1}) 	&=& \mt a_{0}.
			\end{eqnarray*}
			We also let $\varphi$ act on words of length greater than 1, by the morphism property
			\[
				\varphi(uv)=\varphi(u)\varphi(v).
			\]
			Let $\varepsilon$ be the empty word.
			The $k$-bonacci word $W_n=W_n^{(k)}$ is defined by
			\begin{eqnarray*}
				W_n&=&\varepsilon,\quad 0\le n\le k-2,\\
				W_{k-1}&=&\mt a_{k-1},\\
				W_n&=&\varphi_k(W_{n-1}),\quad n\ge k.
			\end{eqnarray*}
		\end{df}
		\begin{lem}
			The length of the $n$th $k$-bonacci word is equal to the $n$th $k$-bonacci number: $\abs{W_n^{(k)}}=w_n^{(k)}$.
		\end{lem}
		\begin{lem}
			For all $k\ge 2$ and $n\ge 2k-1$, if $W_n$ is the $n$th $k$-bonacci word then
			\[
				W_n = W_{n-1}\dots W_{n-k}.
			\]
		\end{lem}
		We then say Fibonacci for 2-bonacci and Tribonacci for 3-bonacci.
		Thus the finite Tribonacci words $T_n$ are defined by 
		\begin{eqnarray*}
			T_0&=&\varepsilon\\
			T_1&=&\varepsilon\\
			T_2&=&\mt{2}\\
			T_3&=&\mt{0}\\
			T_4&=&\mt{01}\\
			T_n &=& T_{n-1}T_{n-2}T_{n-3},\quad n\ge 5,
		\end{eqnarray*}
		and the finite Fibonacci words $F_n$ by
		\begin{eqnarray*}
			F_0&=&\varepsilon\\
			F_1&=&\mt{1}\\
			F_2&=&\mt{0}\\
			F_n &=& F_{n-1}F_{n-2},\quad n\ge 3.
		\end{eqnarray*}

	\begin{df}
		The infinite $k$-bonacci word $\mathbf{f}^{(k)}$ is the fixed point $\varphi_k^{(\infty)}(\mt 0)$ of the morphism $\varphi_k$.
	\end{df}
	Thus, the infinite Tribonacci word
	\[
		T=\mathbf{f}^{(3)}=\mt{0102010010201}\dots
	\]
	is a fixed point of the morphism $\mt{0}\to \mt{01}$, $\mt{1}\to \mt{02}$, $\mt{2}\to \mt{0}$.
	It is a variant of the Fibonacci word obtained from the morphism $\mt{0}\to \mt{01}, \mt{1}\to \mt{0}$.

	\section{Lower bounds from critical exponents}
		A word is square-free if it does not contain any subword of the form $xx$ (denoted $x^2$), $\abs{x}>0$.
		Shallit and Wang showed that a square-free word has high automatic complexity, and we shall show that integrality of powers is
		not crucial: that is, we shall use critical exponents.
		\begin{df}\label{critExp}
			Let $\mathbf{w}$ be an infinite word over the alphabet $\Sigma$, and let $x$ be a finite word over $\Sigma$.
			Let $\alpha>0$ be a rational number.
			The word $x$ is said to occur in $\mathbf{w}$ with exponent $\alpha$ if
			there is a subword $y$ of $\mathbf{w}$ with $y = x^a x_0$ where
			$x_0$ is a prefix of $x$,
			$a$ is the integer part of $\alpha$, and
			$\abs{y}=\alpha \abs{x}$. We say that $y$ is an  $\alpha$-power.
			The word $\mathbf{w}$ is  $\alpha$-power-free if it contains no subwords which are  $\alpha$-powers.

			The critical exponent for $\mathbf{w}$ is
			the supremum of the  $\alpha$ for which $\mathbf{w}$ has  $\alpha$-powers, or equivalently
			the infimum of the  $\alpha$ for which $\mathbf{w}$ is  $\alpha$-power-free.
		\end{df}
		\begin{df}\label{rate}
			Fix a finite alphabet $\Sigma$.
			For an infinite word $\mathbf{w}\in\Sigma^{\infty}$, let $\mathbf{w}\restrict n$ denote the prefix of $\mathbf{w}$ of length $n$.
			Let $C:\Sigma^*\to\mathbb N$.
			The \emph{lower $C$-complexity rate} of $\mathbf{w}$ is
			\[
				\underline{C}(\mathbf w)=\liminf_{n\to\infty}\frac{C(\mathbf{w}\restrict n)}n.
			\]
			The \emph{upper $C$-complexity rate} of $\mathbf{w}$ is
			\[
				\underline{C}(\mathbf w)=\limsup_{n\to\infty}\frac{C(\mathbf{w}\restrict n)}n.
			\]
			If these are equal we may speak simply of the \emph{$C$-complexity rate}.
			In the case where $C=\mathrm{A_N}$ we may speak of \emph{automatic complexity rate}.
		\end{df}
		\begin{df}[Fibonacci constant]
			Let $\phi=\frac{1+\sqrt 5}2$, the positive root of $\phi^2=\phi+1$.
		\end{df}
		\begin{df}[Tribonacci constant]\label{tribConst}
			Let
			\[
				\xi = \frac13\left(
					1+\sqrt[3]{19-3\sqrt{33}}
					 +\sqrt[3]{19+3\sqrt{33}}
				\right).
			\]
		\end{df}
		\begin{thm}[{\cite{MR1170322}}]\label{three6}
			The critical exponent of the infinite Fibonacci word $\mathbf{f}^{(2)}$ is $2+\phi\approx 3.6$.
		\end{thm}
		
		Tan and Wen~\cite{MR2339496} studied critical exponents, calling them \emph{free indices}.
		\begin{thm}[Tan and Wen {\cite[Theorem 4.5]{MR2339496}}]\label{critTrib}
			The critical exponent of the infinite Tribonacci word $\mathbf{f}^{(3)}$
			is $3+\frac12(\theta^2+\theta^4)$,
			where
			\begin{eqnarray*}
				\theta&=& \frac13 \left(-1 - \frac2{(17 + 3 \sqrt{33})^{1/3}} + (17 + 3 \sqrt{33})^{1/3}\right)\\
				&\approx& 0.543689012692076\dots
			\end{eqnarray*}
			is the unique real root of the equation $\theta^3+\theta^2+\theta=1$.
		\end{thm}
		\begin{lem}\label{realZero}
			The critical exponent of $\mathbf{f}^{(3)}$ is the real zero
			\begin{eqnarray*}
				&&2 + \frac16 \sqrt[3]{54 - 6 \sqrt{33}} + \frac{\sqrt[3]{9 + \sqrt{33}}}{6^{2/3}}\\
				&=& 3.19148788\dots
			\end{eqnarray*}
			of the polynomial $2x^3 - 12x^2 + 22x - 13$.
		\end{lem}
		Lemma~\ref{realZero} follows from Theorem~\ref{critTrib} by computer software (Wolfram Alpha).
		\begin{df}
			Let $x$ be a word of length $n$, $x=x_1,\dots,x_n$.
			Two occurrences of words $a$ (starting at position $i$) and $b$ (starting at position $j$) in a word $x$ are 
			\emph{disjoint} if $x=uavbw$ where $u,v,w$ are words and $|u|=i-1$, $|uav|=j$.
			If in addition $|v|>0$ then we say that these occurrences of $a$ and $b$ are \emph{strongly disjoint}.
		\end{df}
		\begin{thm}[{\cite{kjoshanssen2019incompressibility}}]\label{May 4, 2019}
			If the critical exponent of a word $x$ is at most $\gamma\ge 2$ then
			there is an $m\ge 0$ and a set of $m$ many strongly disjoint at-least-square powers in $x$ with
			$\mathrm{A_N}(x)\ge \frac{n+1-m}{\gamma}$.
		\end{thm}
		\begin{thm}\label{September 6, 2019}
			If the critical exponent of a word $x$ is at most $\gamma\ge 2$ then
			$\mathrm{A_N}(x)\ge \frac{n+1-\sqrt{2n}}{\gamma}$.
		\end{thm}
		\begin{proof}
			By uniqueness of path the $m$ many powers in Theorem~\ref{May 4, 2019} must have distinct base lengths.
			Thus the base lengths add up to at least $\sum_{k=1}^m k=m(m+1)/2$, which implies $m(m+1)/2\le n$.
			Consequently $m\le\sqrt{2n}$ and
			\[
				\mathrm{A_N}(x)\ge \frac{n+1-m}{\gamma} \ge \frac{n+1-\sqrt{2n}}{\gamma}.\qedhere
			\]
		\end{proof}
		\begin{thm}
			The $\mathrm{A_N}$-complexity rate of the infinite Fibonacci word $\mathbf{f}^{(2)}$ is at least 
			\[
				\frac2{5+\sqrt 5}
				= 0.27639\dots
			\]
			The $\mathrm{A_N}$-rate of the infinite Tribonacci number $\mathbf{f}^{(3)}$ is at least 
			\[
				0.313333478\dots
			\]
			the real root of $-2 + 12 x - 22 x^2 + 13 x^3$.
		\end{thm}
		\begin{proof}
			These two facts now follow by applying Theorem~\ref{September 6, 2019} with
			Theorems~\ref{three6} and~\ref{critTrib}, respectively.
		\end{proof}

		Karhum{\"a}ki~\cite{MR690339} showed that the Fibonacci words contain no 4th power and this implies
		(the deterministic version is in Shallit and Wang 2001~\cite[Theorem 9]{MR1897300}) $\mathrm{A_N}(x)\ge (n+1)/4$.
		\begin{thm}
			The $\mathrm{A_N}$-complexity rate of the infinite $k$-bonacci word $\mathbf{f}^{(k)}$ is at least $1/4$ for any $k\ge 2$.
		\end{thm}
		\begin{proof}
			Glen~\cite{MR2331002} showed that the $k$-bonacci word has no fourth power for any $k\ge 2$.
			Thus, the critical exponent is at most 4 and by Theorem~\ref{September 6, 2019} we are done.
		\end{proof}

	\section{Upper bounds from factorizations}
		There are many factorization result possible for $k$-bonacci words. Even their definitions like $F_n = F_{n-1}F_{n-2}$
		are factorizations. We shall prove some such results that help us obtain upper bounds on automatic complexity:
		Theorem~\ref{prefix} and~\ref{also-prove}.
		In the following, for convenience we renumber by defining $\tilde T_n = T_{n+3}$.
		\begin{df}
			For $n\ge 0$ and $0\le k\le n$, we let $\ci{k}_n = \tilde T_{n-k}$.
			We also write $\di{k}_n=|\ci{k}_n|$, the length of $\ci{k}_n$.
			When $n$ is understood from context we write $\ci{k} = \ci{k}_n$ and $\di{k}=\di{k}_n$.
		\end{df}

		\begin{thm}\label{prefix}
			For large enough $n$,
			$\tilde T_{n-2}^2 \prod_{k=6}^{3\lfloor n/3\rfloor+1} \tilde T_{n-k}$ is a prefix of $\tilde T_n$.
		\end{thm}
		\begin{proof}
			Note that the equation $\tilde T_n=\tilde T_{n-1}\tilde T_{n-2}\tilde T_{n-3}$ holds for $n\ge 2$.\footnote{
				It does not hold for $n=1$: we have
				\(
					\tilde T_2=(\mt{01})(\mt{0})(\mt{2})=\tilde T_1\tilde T_0\tilde T_{-1},
				\)
				but
				\(
					\tilde T_1=\mt{01}\ne (\mt{0})(\mt{2})() = \tilde T_{0}\tilde T_{-1}\tilde T_{-2}.
				\)
			}
			Thus we can write
			\begin{equation}\label{limiting}
				\ci{n-2}=\ci{n-1}\ci{n}\ci{n+1}
			\end{equation}
			but we cannot expand $\ci{n-1}$.
			The idea now is to use a loop of length 13 followed by one of length 6:
			\begin{eqnarray*}
				\tilde T_n &=& \ci{1} \ci{2} \ci{3} = (\ci{2}\ci{3}\ci{4}) \ci{2}\ci{3}
				= \ci{2}\ci{3}\ci{4}(\ci{3}\ci{4}\ci{5})\ci{3}
			\\
				&=&\ci{2}\ci{3}\ci{4}(\ci{4}\ci{5}\ci{6})\ci{4}\ci{5}\ci{3}
				=\ci{2}\ci{3}\ci{4}(\ci{5}\ci{6}\ci{7})\ci{5}\ci{6}\ci{4}\ci{5}\ci{3}
			\\
				&=&\ci{2}^2\ci{6}\ci{7}\enskip \ci{5}\enskip\ci{6}\ci{4}\ci{5}\enskip\ci{3}
			\\
				&=&\ci{2}^2\ci{6}\ci{7}\enskip (\ci{8}\ci{9}\ci{10} \ci{8}\ci{9} \ci{7}\ci{8})\enskip \ci{6}\ci{4}\ci{5}\enskip\ci{3}
			\\
				&=&\ci{2}^2\left(\prod_{k=6}^{10}\ci{k}\right) \ci{8}\enskip\ci{9} \ci{7}\ci{8}\enskip \ci{6}\ci{4}\ci{5}\enskip\ci{3}
			\\
				&=&\ci{2}^2\left(\prod_{k=6}^{13}\ci{k}\right)\enskip
				\ci{11}\ci{12}\ci{10}\ci{11} \ci{9} \ci{7}\ci{8} \ci{6}\ci{4}\ci{5}\ci{3}.
			\end{eqnarray*}
			Thus for $m=4$ we have
			\begin{equation}\label{circ}
				\ci{0}=\ci{2}^2 \left(\prod_{k=6}^{3m+1}\ci{k}\right) \ci{3m-1}\left(\prod_{M=m}^0 \ci{3M}\ci{3M-2}\ci{3M-1}\right) \ci{3}.
			\end{equation}
			Here we use the notation $\prod_{M=m}^0 a_M = a_M a_{M-1}\dots a_0$.
			To prove (\ref{circ}) for $m\ge 4$ by induction we expand:
			\[
				\ci{3m-1}=\ci{3m+2}\ci{3m+3}\ci{3m+4}\ci{3m+2}\ci{3m+3}\ci{3m+1}\ci{3m+2}
			\]
			This is valid as long as $3m+1\le n-2$ by (\ref{limiting}), i.e., as long as $3(m+1)\le n$.
		\end{proof}
		\begin{thm}\label{also-prove}
			For any $m\ge 4$,
			\(
				\prod_{k=6}^{3m+1}\ci{k}
			\)
			is a prefix of $\ci{2}$.
		\end{thm}
		\begin{proof}
			We have
			\[
				\ci{2}=\ci{3}\ci{4}\ci{5}=\ci{4}\ci{5}\ci{6}\ci{4}\ci{5}
				= \ci{6}\enskip \ci{7}\ci{8}\enskip\ci{6}\ci{7}\enskip\ci{5}\ci{6}\enskip\ci{4}\ci{5}
			\]
			By substitution into the proof it will suffice to show that $\ci{4}\ci{5}\dots$ is a prefix of $\ci{0}$.

			To show that, we first show that $\ci{3}\ci{4}\dots$ is a prefix of $\ci{0}$:
			\begin{eqnarray*}
				\ci{0}&=&\ci{1}\ci{2}\ci{3}=\ci{2}\ci{3}\ci{4}\ci{2}\ci{3}
			\\
				&=&\ci{3}\ci{4}\ci{5}\enskip\ci{3}\enskip\ci{4}\ci{2}\ci{3}
			\end{eqnarray*}
			Now by substitution, $\ci{6}\ci{7}\ci{8}\dots$ is a prefix of $\ci{3}$ and we are done.
			(Incidentally this can now be used to show that $\ci{2}\ci{3}\dots$ is a prefix of $\ci{0}$.)
			Finally, let us show that $\ci{4}\ci{5}\dots$ is a prefix of $\ci{0}$:
			\begin{eqnarray*}
				\ci{0}&=&\ci{1}\ci{2}\ci{3}=\ci{2}\ci{3}\ci{4}\ci{2}\ci{3}
				=\ci{3}\ci{4}\ci{5}\ci{3}\ci{4}\ci{2}\ci{3}
			\\
				&=&\ci{4}\ci{5}\ci{6}\enskip\ci{4}\ci{5}\enskip\ci{3}\ci{4}\enskip\ci{2}\ci{3}
			\end{eqnarray*}
			By substitution and Theorem~\ref{prefix}, $\ci{7}\ci{8}\ci{9}\dots$ is a prefix of $\ci{4}$, and we are done.
		\end{proof}
		We can also expand the tail of $\tilde T_n$:
		\[
			\ci{0} = \ci{1}\ci{3}\ci{4}\ci{5}\ci{4}\ci{5}\ci{6}
		\]
		which ends in $(\ci{4}\ci{5})^*$ with $*$ a non-integer exponent.
		\begin{thm}
			For $n\ge 6$,
			\[
				\mathrm{A}^-(\tilde T_n)\le \di{1}-\sum_{k=6}^{3\lfloor n/3\rfloor+1}\di{k}.
			\]
		\end{thm}
		For example, when $n=6$ this is $24-(1+1)=22$, which fits with an 8-day long computation we performed.
		\begin{proof}
			Using now a $\di{2}$-cycle followed by a path and then a $\di{4}+\di{5}$-cycle,
			we can subtract the extra prefix from Theorem~\ref{prefix} and use only
			\begin{eqnarray*}
				\di{2}& + &\di{4}+\di{5} + \left(\di{0}-2\di{2}-2(\di{4}+\di{5})-\di{6}-\sum_{k=6}^{3\lfloor n/3\rfloor+1}\di{k}\right)
			\\
				= \di{1}&-&\sum_{k=6}^{3\lfloor n/3\rfloor+1}\di{k}
			\end{eqnarray*}
			(since $\di{0}=\di{6}+2\di{5}+3\di{4}+2\di{3}+\di{2}$)
			states. Uniqueness is guaranteed by Theorem~\ref{uniqueish}.
		\end{proof}
		\begin{lem}[{\cite{MR2003519}}]
			The Tribonacci constant (Definition~\ref{tribConst}) satisfies
			\[
				\xi=\lim_{n\to\infty}\frac{\di{1}_n}{\di{0}_n}=1.83929\dots
			\]
			and is the unique real root of $\xi^3=\xi^2+\xi+1$.
		\end{lem}
		In particular $\xi=1/\theta$ with $\theta$ as in Theorem~\ref{critTrib}.
		\begin{thm}\label{uniqueish}
			For large enough $n$, the equation
			\begin{equation}\label{hara}
				x\di{2}+y(\di{4}+\di{5})=2(\di{2}+\di{4}+\di{5})
			\end{equation}
			for nonnegative integers $x$, $y$, has the unique solution $x=y=2$.
		\end{thm}
		\begin{proof}
			Suppose $(x, y)$ is a solution, not equal to $(2,2)$.
			Then we have $x=0$, $x=1$, $y=0$, or $y=1$.
			If $x=0$ then
			\[
				\mathbb Z\ni y=2\left(\frac{\di{2}+\di{4}+\di{5}}{\di{4}+\di{5}}\right)
				=2\left(\frac{\di{6}}{\di{4}+\di{5}} + 3\right)
			\]
			which is impossible as soon as $\di{6}>0$ since $2\di{6}<\di{4}+\di{5}$.
			Similarly, if $x=1$ then
			\[
				\mathbb Z\ni y = \frac{2(\di{2}+\di{4}+\di{5})-\di{2}}{\di{4}+\di{5}}
			\]
			\[
				= \frac{\di{2}}{\di{4}+\di{5}}+2 = \frac{\di{6}}{\di{4}+\di{5}}+4.
			\]
			If $y=0$ then
			\[
				\mathbb Z\ni x=2\frac{\di{2}+\di{4}+\di{5}}{\di{2}} = 2\left(1+\frac{\di{2}-\di{3}}{\di{2}}\right) = 2\left(2-\frac{\di{3}}{\di{2}}\right)
			\]
			which is impossible as soon as $\di{3}>0$ since
			\[
				2>\frac{2\di{3}}{\di{2}} = \frac{\di{3}+\di{4}+\di{5}+\di{6}}{\di{3}+\di{4}+\di{5}}>1.
			\]
			Finally, if $y=1$ then
			\[
				\mathbb Z\ni x = \frac{2(\di{2}+\di{4}+\di{5})-\di{4}-\di{5}}{\di{2}}=2+\frac{\di{4}+\di{5}}{\di{2}}
			\]
			which is impossible as soon as $0<\di{4}+\di{5}<\di{2}$, i.e., $\di{4}>0$.
		\end{proof}
		\begin{thm}
			The $\mathrm{A}^-$-complexity rate of the Tribonacci word satisfies
		\begin{eqnarray*}
			\limsup_{n\to\infty}\frac{\mathrm{A}^-(T_n)}{|T_n|}
			&\le&1/6 (-8 + (586 - 102 \sqrt{33})^{1/3} + (2 (293 + 51 \sqrt{33}))^{1/3})\\
			&\approx& 0.4870856\dots
		\end{eqnarray*}
		\end{thm}
		\begin{proof}
			We calculate
		\begin{eqnarray*}
			\limsup_{n\to\infty}\frac{\mathrm{A}^-(\tilde T_n)}{|\tilde T_n|}&\le&
			\lim_{n\to\infty}\frac{\di{1}-\sum_{k=6}^{3\lfloor n/3\rfloor+1}\di{k}}{\di{0}}
			= \frac1{\xi} - \sum_{k=6}^\infty \frac1{\xi^k}\\
			&=&\frac1{\xi}-\frac1{3\xi^2+3\xi+2}\\
			&=& \frac16 (-8 + (586 - 102 \sqrt{33})^{1/3} + (2 (293 + 51 \sqrt{33}))^{1/3}).\qedhere\\
		\end{eqnarray*}
		\end{proof}
		\begin{df}
			$\mathrm{A_N^{lower}}(x)$ is the minimal $q$ such that for all sequences of strongly disjoint powers
			\[
				x_1^{\alpha_1},\dots,x_m^{\alpha_m},
			\] in $x$, with the uniqueness condition
			that
			\[
				\sum \alpha_i |x_i|=\sum \alpha_i y_i \implies y_i=|x_i|,\text{all }i,
			\]
			we have
			\begin{equation}\label{fundamental}
				2q\ge n+1-m-\sum_{i=1}^m (\alpha_i-2) |x_i|.
			\end{equation}
		\end{df}
		The definition of $\mathrm{A_N^{lower}}(x)$ may seem very technical.
		The point is that
		\begin{itemize}
			\item $\mathrm{A_N^{lower}}$ appears to be faster to compute than $\mathrm{A_N}$,
			\item by~\cite[Theorem 19]{kjoshanssen2019incompressibility}, we have $\mathrm{A_N^{lower}}(x)\le \mathrm{A_N}(x)$ for all words $x$, and
			\item $\mathrm{A_N^{lower}}(x)$ is a better lower bound than that obtained simply by the critical exponent considerations in Theorem~\ref{May 4, 2019}.
		\end{itemize}
		We have implemented $\mathrm{A_N}$ and $\mathrm{A_N^{lower}}$ (\cite{KF}) with results in Table~\ref{resultsA_N} and Table~\ref{resultsA_Nlower}.
		(A lookup tool is also available for automatic complexity~\cite{lookup}.)
		Note that 
		\begin{eqnarray}\label{notWide}
			\mathrm{A}^-(T_n)
			&=&\mathrm{A}^-(\tilde T_{n-3})\le \di{1}-\sum_{k=6}^{3\lfloor (n-3)/3\rfloor+1}\di{k}\nonumber\\
			&=&\di{1}-\sum_{k=6}^{3\lfloor n/3\rfloor-2}\di{k}
		\\
			&=&\di{1}-\sum_{k=6}^{7}\di{k} = t_8 - t_3 - t_2 = 24 - 1 - 1 = 22,\quad (n=9),
		\nonumber\\
			&=&\di{1}-\sum_{k=6}^{7}\di{k} = t_9 - t_4 - t_3 = 44 - 2 - 1 = 41,\quad (n=10).\nonumber
		\end{eqnarray}
		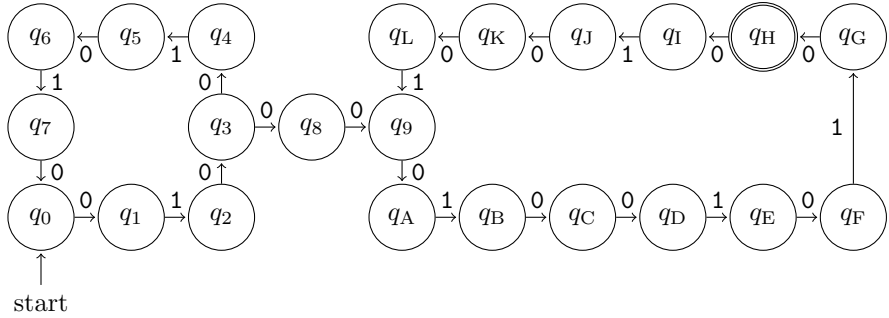
\begin{figure}
			\[
				\begin{tikzpicture}[shorten >=1pt,node distance=1.2cm,on grid,auto]
					\node[state,initial below] (q_0) {$q_0$};
					\node[state] (q_1) [right=of q_0] {$q_1$};
					\node[state] (q_2) [right=of q_1] {$q_2$};
					\node[state] (q_3) [above=of q_2] {$q_3$};
					\node[state] (q_4) [above=of q_3] {$q_4$};
					\node[state] (q_5) [left=of q_4] {$q_5$};
					\node[state] (q_6) [left=of q_5] {$q_6$};
					\node[state] (q_7) [below=of q_6] {$q_7$};
					\node[state] (q_8) [right=of q_3] {$q_8$};
					\node[state] (q_9) [right=of q_8] {$q_9$};
					\node[state] (q_A) [below=of q_9] {$q_\mathrm{A}$};
					\node[state] (q_B) [right=of q_A] {$q_\mathrm{B}$};
					\node[state] (q_C) [right=of q_B] {$q_\mathrm{C}$};
					\node[state] (q_D) [right=of q_C] {$q_\mathrm{D}$};
					\node[state] (q_E) [right=of q_D] {$q_\mathrm{E}$};
					\node[state] (q_F) [right=of q_E] {$q_\mathrm{F}$};
					\node (place) [above=of q_F] {};
					\node[state] (q_G) [above=of place] {$q_\mathrm{G}$};
					\node[state,accepting] (q_H) [left=of q_G] {$q_\mathrm{H}$};
					\node[state] (q_I) [left=of q_H] {$q_\mathrm{I}$};
					\node[state] (q_J) [left=of q_I] {$q_\mathrm{J}$};
					\node[state] (q_K) [left=of q_J] {$q_\mathrm{K}$};
					\node[state] (q_L) [left=of q_K] {$q_\mathrm{L}$};
					\path[->] 
						(q_0) edge node {$\mt{0}$} (q_1)
						(q_1) edge node {$\mt{1}$} (q_2)
						(q_2) edge node {$\mt{0}$} (q_3)
						(q_3) edge node {$\mt{0}$} (q_4)
						(q_4) edge node {$\mt{1}$} (q_5)
						(q_5) edge node {$\mt{0}$} (q_6)
						(q_6) edge node {$\mt{1}$} (q_7)
						(q_7) edge node {$\mt{0}$} (q_0)
						(q_3) edge node {$\mt{0}$} (q_8)
						(q_8) edge node {$\mt{0}$} (q_9)
						(q_9) edge node {$\mt{0}$} (q_A)
						(q_A) edge node {$\mt{1}$} (q_B)
						(q_B) edge node {$\mt{0}$} (q_C)
						(q_C) edge node {$\mt{0}$} (q_D)
						(q_D) edge node {$\mt{1}$} (q_E)
						(q_E) edge node {$\mt{0}$} (q_F)
						(q_F) edge node {$\mt{1}$} (q_G)
						(q_G) edge node {$\mt{0}$} (q_H)
						(q_H) edge node {$\mt{0}$} (q_I)
						(q_I) edge node {$\mt{1}$} (q_J)
						(q_J) edge node {$\mt{0}$} (q_K)
						(q_K) edge node {$\mt{0}$} (q_L)
						(q_L) edge node {$\mt{1}$} (q_9)
						;
				\end{tikzpicture}
			\]
			\caption{An automaton witnessing the automatic complexity of the Fibonacci word of length 55.}\label{japan}
		\end{figure}
		\begin{rem}
			Witnessing automata for $\mathrm{A_N}$ are conveniently generated by \emph{state sequences}.
			A state sequence is the sequence of states visited by the unique accepting path of length $n+1$
			(having potentially up to $n$ edges and $n+1$ states).
			A week-long computer search for the length 55 Fibonacci word
			\[
				\mt{01001010010}
				\mt{01010010100}
				\mt{10010100100}
				\mt{10100101001}
				\mt{00101001010}
			\]
			revealed the witnessing state sequence, where states are given numerical labels, using letters $\mathrm{A,B,C},\dots$ for the numbers $10,11,12,\dots$:
			\begin{eqnarray*}
				 \mathrm{0,  1,  2,  3,  4,  5,  6,  7,  0,  1,  2,  3,  4,  5,  6,  7,  0,  1, 2, 3, 8, 9, A, B, C, D, E, F, G, H, }
			\\
				\mathrm{I, J, K, L,  9, A, B, C, D, E, F, G, H, I, J, K, L,  9, A, B, C, D, E, F, G, H.}
			\end{eqnarray*}
			We illustrate the automaton induced by this state sequence in Figure~\ref{japan}.
			Generalizing this example gives
			\[
				\overline{\mathrm{A_N}}(\mathbf{f}^{(2)})\le
				\frac1{\varphi^2}+\frac1{\varphi^7}=0.41640786499
			\]
		\end{rem}
		\begin{lem}\label{uniqueSol}
			For $n\ge 6$, the equation
			\[
				x f_{n-2} + y f_{n} = 2(f_{n-2}+f_{n})
			\]
			for nonnegative integers $x$, $y$, has the unique solution $x=y=2$.
		\end{lem}
		\begin{proof}
			If $x=0$ then
			\[
				y=2\frac{f_{n-2}+f_n}{f_n}=2+\frac{2f_{n-2}}{f_n} = 2+\frac{2f_{n-2}}{f_{n-1}+f_{n-2}}\in (2,3)
			\]
			is not an integer as long as $n\ge 4$.
			If $x=1$ then
			\[
				y=\frac{2(f_{n-2}+f_n)-f_{n-2}}{f_n} = 2 + \frac{f_{n-2}}{f_n} \in (2,3)
			\]
			is not an integer as long as $n\ge 3$.
			If $y=0$ then
			\[
				x=2\,\frac{f_{n-2}+f_n}{f_{n-2}}
				=2\,\frac{3f_{n-2}+f_{n-3}}{f_{n-2}}
				= 6 + \frac{f_{n-3}+f_{n-4}+f_{n-5}}{f_{n-3}+f_{n-4}}\in (7,8)
			\]
			is not an integer as long as $f_{n-5}>0$, i.e., $n\ge 6$.
			If $y=1$ then 
			\[
				x=\frac{2(f_{n-2}+f_n)-f_n}{f_{n-2}} = 2 + \frac{f_n}{f_{n-2}} = 2+\frac{f_{n-2}+f_{n-2}+f_{n-3}}{f_{n-2}} \in (4,5)
			\]
			as long as $f_{n-3}>0$.
		\end{proof}
		\begin{thm}\label{interm}
			The upper automatic complexity rate of the infinite Fibonacci word $\overline{\mathrm{A_N}}(\mathbf{f}^{(2)})$ is at most $\frac2{\varphi^3}$.
		\end{thm}
		\begin{proof}
			We exploit the lengths of Fibonacci words.
			\begin{eqnarray*}
				f_n &=& f_{n-1}+f_{n-2}=f_{n-2}+f_{n-3}+f_{n-2}\\
				&=& f_{n-3}+f_{n-4}+(f_{n-3}+f_{n-3}+f_{n-4})\\
				&=& f_{n-4}+f_{n-5}+f_{n-4}+(f_{n-3}+f_{n-3}+f_{n-4})\\
				&=& \underbrace{f_{n-4}}_{\text{hardcode this}}
				+\underbrace{(f_{n-5}+f_{n-5}+f_{n-6})}_{\text{$f_{n-5}$-state cycle}}
				+\overbrace{(f_{n-3}+f_{n-3}+f_{n-4})}^{\text{$f_{n-3}$-state cycle}}.
			\end{eqnarray*}
			This way we obtain for a Fibonacci word $x$ of length $f_n$ that
			\[
				\mathrm{A_N}(x)\le f_{n-4}+f_{n-5}+f_{n-3}=2f_{n-3}.
			\]
			In the limit, $f_n/f_{n-1}\sim\varphi=\frac{1+\sqrt 5}2\approx 1.6$, so $f_n/f_{n-3}\sim \varphi^3=4.236$,
			so
			\[
				\mathrm{A_N}(x)\le \frac2{\varphi^3}f_n = 0.47f_n.
			\]
			The cycles give a unique path of length $f_n$ for large enough $n$, since
			\[
				f_{n-5}x+f_{n-3}y=2(f_{n-5}+f_{n-3})
			\]
			has a unique solution $x=y=2$ by Lemma~\ref{uniqueSol}.
		\end{proof}
		\begin{table}
			\begin{tabular}{l l l l l l |l| l}
				\toprule
				$n$ &$t_n$	&$T_n$	& .313$t_n$ & $\mathrm{A_N^{lower}}$	&$\mathrm{A}^-$	&(\ref{notWide}) &$.487t_n$\\
				\midrule
				0	& 0		&		&0			&1							&1		&		&0\\
				1	& 0		&		&0			&1							&1		&		&0\\
				2	& 1		&$\mt{2}$		&.3			&1							&1		&		&0.49\\
				3	& 1		&$\mt{0}$		&.3			&1							&1		&		&0.49\\
				4	& 2		&$\mt{01}$		&.6			&2							&2		&		&0.97\\
				5	& 4		&$\mt{0102}$	&1.3		&3							&3		&		&1.95\\
				6	& 7&$\mt{0102010}$		&2.2		&4							&4		&		&3.4\\
				7	& 13&$\mt{0102010010201}$&4.1		&7 							&7		&		&6.3\\
				8	& 24&$\mt{0102}\dots \mt{0100102}$&7.5	&12						&13		&		&11.7\\
				9	& 44&$\mt{0102}\dots \mt{0102010}$&13.8	&21						&		&22		&21.4\\
				10	& 81&$\mt{0102}\dots \mt{0010201}$&25.4	&36						&		&41		&39.4\\
				\bottomrule
			\end{tabular}
			\caption{Lower and limiting upper bounds on $\mathrm{A_N}(T_n)$ and $\mathrm{A}^-(T_n)$.
		}\label{resultsA_N}
		\end{table}
		\begin{table}
			\begin{tabular}{l l l l l l}
				\toprule
				$n$ &$f_n$	& $F_n$ 						& $.276f_n$ & $\mathrm{A_N^{lower}}$ & $.382f_n$\\
				\midrule
				0	&0		&								&0	&1&0\\
				1	&1		&$\mt{1}$					&.3	&1&0.4\\
				2	&1		&$\mt{0}$					&.3	&1&0.4\\
				3	&2		&$\mt{01}$					&.6	&2&0.8\\
				4	&3		&$\mt{010}$					&.8	&2&1.1\\
				5	&5		&$\mt{01001}$				&1.4	&3&1.9\\
				6	&8		&$\mt{01001010}$			&2.2	&4&3.1\\
				7	&13		&$\mt{0100101001001}$		&3.6		&6&5.0\\
				8	&21		&$\mt{010010100100101001010}$&5.8&9&8.0\\
				9	&34		&$\mt{010010100100101001010\dots 1001001}$&9.4&14&13.0\\
				10	&55		&$\mt{0100}\dots \mt{1010}$&15.2& 21&21.0\\
				\bottomrule
			\end{tabular}
		\caption{Lower and limiting upper bounds on $\mathrm{A_N}(F_n)$.}\label{resultsA_Nlower}
		\end{table}

		\section{Conclusion}
			More can be done on automatic complexity rates of $k$-bonacci words.
			For instance, we conjecture that $\overline{\mathrm{A_N}}(\mathbf{f}^{(2)})\le 1/\varphi^2=.382$.
			More precisely, we conjecture that this can be shown by analyzing the decomposition
			\begin{eqnarray*}
				\ci{0} &=& \ci{1}\ci{2} = \ci{2}\ci{3}\ci{2} = \ci{3}\ci{4} \ci{3} \ci{3}\ci{4} = \ci{4}\ci{5} \ci{4} \ci{4}\ci{5} \ci{4}\ci{5} \ci{4}
			\\
				&=& \ci{5}\ci{6}\ci{5} \ci{5}\ci{6}\ci{5} \ci{6} \ci{5} \ci{5}\ci{6} \ci{5} \ci{5}\ci{6}
			\\
				&=& \left(\ci{6}\ci{7}\ci{6}\ci{5}\right) \left(\ci{6}\ci{7}\ci{6}\ci{5}\right) \left(\ci{6}\right) \ci{5} \ci{5} \ci{6} \ci{6}\ci{7} \ci{6}\ci{7}\ci{6}.
			\end{eqnarray*}
			However, in the present article we are content to have proven that the Fibonacci word has intermediate automatic complexity rate in Theorem~\ref{interm}.

	\bibliographystyle{plain}
	\bibliography{tribonacci}
\end{document}